\documentclass{amsart}
\usepackage{hyperref}
\usepackage{tikz}
\usetikzlibrary{matrix,arrows}

\author{Robert Kesler and Benjamin Steinhurst}
\title{The Casimir effect on Laakso spaces}

\newcommand{\B}[1]{\ensuremath{\mathbb{#1}}}

\newtheorem{theorem}{Theorem}[section]
\newtheorem{cor}{Corollary}[section]

\newtheorem{prop}{Proposition}[section]

\newtheorem{definition}{Definition}[section]

\numberwithin{equation}{section}
\date{9 July 2012}
\begin{document}

\begin{abstract}
We explore the properties of an analog to the Casimir effect on Laakso spaces such as the dependence on the separation of the plates and boundary effects. We also mention some results on the influence of complex poles in the spectral zeta function over finite approximations to Laakso spaces. 



\end{abstract}

\maketitle

\begin{center}{\small
Contact:\\
rkesler@math.cornell.edu\\
Department of Mathematics\\
Cornell University\\
Ithaca~NY~14850~USA\\
\ \\
steinhurst@math.cornell.edu\\
Department of Mathematics\\
Cornell University\\
Ithaca~NY~14850~USA}\end{center}

\section{Introduction}


With recent advances in fractal analysis there has become available a large amount of information concerning Laplacians and their spectra over fractal spaces \cite{Kig01,Strichartz2006,BegueEtAl2012}. There has been progress in using this information to construct and analyze analogs to physical systems, e.g. the behavior of a photon in a fractal \cite{ADT2010} and other ``physical'' consequences \cite{ADT2009,Dunne2012,FKS09,S09}. The physical consequence of fractal geometry that we explore in this paper is the Casimir effect \cite{Spruch1986,MartinBuenzli2006}.

In these works, the underlying space is typically a finitely ramified fractal with a symmetry condition; however, the spectrum of the Laplacian for these objects is generally not known exactly or only described as a scaled Julia set, which means the growth estimates for the eigenvalue counting functions must be used instead. Laakso spaces, whose exact spectrum the authors previously computed, enable us to avoid this complication. In \cite{KKPSS} the authors also computed the exact eigenfunctions of the Hamiltonian with a square well potential, the spectral zeta function for certain defining sequences $\{ j_n\}$, and a Casimir effect on a 1 dimensional arrangement. This paper continues the analysis of the Casimir effect on Laakso spaces.


We begin with defining Laakso spaces in a convenient manner in Section \ref{sec:laakso} where we will also give an explicit description of the spectrum of a natural Laplacian on Laakso spaces. Following this is Section \ref{sec:zeta} where we discuss the general properties and calculations for  spectral zeta functions over Laakso spaces. In Section \ref{sec:finite} we observe that the complex dimensions appear in a model which has only finite complexity and is in principle constructible as a physical object. In Section \ref{sec:CasimirL} we revisit the authors' earlier work in \cite{KKPSS} and determine the strength of the Casimir effect in a Laakso space as a function of both the defining sequence $\{j_i\}$ and the distance between the plates. 
Lasly, in Section \ref{sec:CasimirL} we construct a $3+$ dimensional arrangement involving a Laakso space and show that the Casimir pressure  is proportional to the inverse fourth power of the separation distance. This power has the same exponent as the classical Casimir effect between two parallel uncharged conducting plates; however, what makes this result unusual is that this exponent is not equal to the spectral dimension, $d_s$, plus an integer as commonly seen on fractal domains. This reflects the fact that Laplacians on Laakso spaces are truly $1-$dimensional rather than $d_s-$dimensional operators. The rapid growth in the eigenvalue counting function is due more to the geometry of Laakso spaces and their graph approximations than the nature of the Laplacian.

{\bf Acknowledgements:} We thank Christopher Kauffman, Amanda Parshall, and Evelyn Stamey for their work on the foundational counting arguments that are so often used in this paper. Also Erik Akkermans and Alexander Teplyaev for their frequent and useful comments and challenges. 

\section{Laakso spaces}\label{sec:laakso}

These spaces were introduced in \cite{Laakso2000} and the spectral theory on them developed in \cite{RS2009,KKPSS,SteinhurstTeplyaevA}. We will use the construction indicated in \cite{BE04} and spelled out in detail in \cite{SteinhurstTeplyaevA}. Let $\{j_i\}$ a sequence of integers such that
\begin{equation}
	\lim_{n \rightarrow \infty} \left( \prod_{i=1}^{n} j_i \right)^{1/n} = r.
\end{equation}
The following construction can proceed without this restriction for any integer sequence but the limits that will be taken may not exist otherwise. Define
\begin{equation}
	d_n =  \prod_{i=1}^{n} j_i \hspace{.5in} L_n = \left\{ \frac{m}{d_n} :\ m = 1,\ldots , d_n-1 \right\}.
\end{equation}
Where $L_n$ will be the locations of ``wormholes'' of lever $n$ or lower and $L_n \setminus L_{n-1}$ the locations of the new ``wormholes'' at level $n$. We will write $I = [0,1]$ and $K$ for a Cantor set.

Set $F_0 = I$, $G = \{0,1\}$, and $B_n = L_n \setminus L_{n-1} \subset I$. Let $\phi_{1,0}:F_0 \times G \rightarrow F_0$ be the projection onto the interval $F_0$. Define $F_1 = F_0 \times G / \phi_{1,0}^{-1}(B_1)$. Inductively construct $\phi_{n,n-1}$ and $F_n$. Notice that there are also naturally defined projections $\phi_{n,m}:F_n \rightarrow F_m$. Let $\mu_n$ be the probability measure on $F_n$ that is inherited from Lebesgue measure on $F_0$.

\begin{prop}
The system $(F_n,\phi_{n,n-1},\mu_n)$ is a projective system of measure spaces.
\end{prop}

It is important to note that $F_n$ is a metric graph as it is a collection of line segments of all of length $d_n^{-1}$ at nodes whose locations have a coordinate in $I$ taken from $L_n$. It is this particularly regular structure for the $F_n$ that we will use in the absence of any strict geometric self-similarity. 

\begin{definition}
The projective limit of $F_n$ is a Laakso space with data $\{j_i\}$. This is written as $\lim_{\leftarrow} F_n = L$. There are also associated projections $\Phi_n:L \rightarrow F_n$ such that $\phi_{n,m} \circ \Phi_n  = \Phi_m$ for all $m  \leq  n$.
\end{definition}

For more on projective limits of measure spaces see \cite{HY88}.

\begin{theorem}
For any choice of $\{j_i\}$ such that $r$ exists the corresponding Laakso space $L$ is a complete geodesic metric measure space with Hausdorff dimension $1+\frac{log(r)}{log(2)}$.
\end{theorem}

That the presented construction gives a Laakso space is Lemma 4.6.1 in \cite{Steinhurst2010}. The properties of the Laakso spaces are proved in \cite{Laakso2000}.

Let $\Delta_n$ be the self-adjoint Laplacian on $F_n$ acting as $-\frac{d^2}{dx^2_e}$ where $x_e$ is a coordinate on each line segment. Then the domain of $\Delta_n$ is taken to be the closure of all continuous functions on $F_n$  that are twice differentiable when restricted to each interval and satisfy Kirchoff matching conditions at all vertices. This forces Neumann boundary conditions at the degree one vertices that form the boundary of $F_n$. 

Since $\phi_{n,m}$ maps $F_n$ onto $F_m$ we can by composition use $\phi_{n,m}$ to map functions over $F_m$ to functions over $F_n$ by the convention $\phi_{n,m}^{*} f= f \circ \phi_{n,m}$. The same definition is used for $\Phi_n^{*}$ as well. 

\begin{prop}
For $m < n$, $\phi_{n,m}^{*}Dom(\Delta_m) \subset Dom(\Delta_n)$.
\end{prop} 

\begin{theorem}[\cite{RS2009}]
There exists a self-adjoint Laplacian $\Delta$ on $L$ such that $\Delta\Phi_n^{*}f = \Phi_n^{*}\Delta_n f$ for all $f \in Dom(\Delta_n)$ and for all $n$ with domain 
\begin{equation}
	Dom(\Delta) = \overline{\bigcup_{n=0}^{\infty} \Phi_n^{*}Dom(\Delta_n)}.
\end{equation}
Furthermore
\begin{eqnarray}
	\sigma(\Delta) &=& \bigcup_{k=0}^{\infty} \{\pi^2 k^2\}\cup\bigcup_{n=1}^\infty \bigcup_{k=0}^{\infty} \{(k+1/2)^2\pi^2 d_n^2\}\cup\bigcup_{n=1}^\infty \bigcup_{k=1}^\infty \{k^2\pi^2d_n^2\}\notag\\ &&\cup\bigcup_{n=2}^\infty \bigcup_{k=1}^\infty \{k^2\pi^2d_n^2\} \cup \bigcup_{n=2}^\infty \bigcup_{k=1}^{\infty} \left\{\frac{k^2\pi^2 d_n^2}{4}\right\}
\end{eqnarray}
And multiplicities
\begin{equation}
1,\hspace{.5cm} 2^n,\hspace{.5cm} 2^{n-1}(j_n-2)d_{n-1},\hspace{.5cm} 2^{n-1}(d_{n-1}-1),\hspace{.5cm} 2^{n-2}(d_{n-1}-1)
\end{equation}\label{eqn:multi}
\end{theorem}

The method of calculating the spectrum is based on the fact that an eigenfunction of $\Delta_n$ can be localized between wormholes since interior wormholes are degree four vertices and it is possible for the function to be constant on two of the incoming edges and non-constant on the other two and still satisfy the Kirchoff matching conditions that all incoming first derivatives sum to zero. The spectrum is then determining by breaking down the graphs, $F_n$ into subgraphs on which such eigenfunctions are supported and the multiplicities counted by counting the number of each of these subgraphs. These counting arguments will be revisited later in this paper.

\section{Spectral Zeta Functions}\label{sec:zeta}

\begin{definition}
Let $\Delta$ be a self-adjoint positive-definite operator with a discrete spectrum $\lambda_i$ and multiplicities $g_i$. Then the spectral zeta function is defined, where convergent as
\begin{equation}
	\zeta_{\Delta}(s) = \sum_{i=1}^{\infty} \frac{g_i}{\lambda_i^{s}}.
\end{equation}
We will also denote the analytic continuation of this function as $\zeta_{\Delta}(s)$.
\end{definition}

If one considers the Laplacian on $[0,1]$ with Dirichlet boundary conditions the spectrum is $\{k^{2}\pi^{2}\}_{k=1}^{\infty}$ and all of multiplicity one. Then
\begin{equation}
	\zeta(s) = \sum_{k=1}^{\infty} \frac{1}{(k\pi)^{2s}} = \frac{1}{\pi^{2s}}\zeta_R(2s)
\end{equation}
Where $\zeta_R(s)$ is the Riemann zeta function which is know to have a meromorphic continuation to the whole complex plane. In the case of a Laakso space with data $j_i = 2$
\begin{eqnarray}
	\zeta_L(s)= &&\frac{\zeta_R(2s)}{\pi^{2s}}\left( 
\frac{4(2^{2s-1}+1)}{4^{s}(4^{2}-4)}+\frac{6(2^{2s-1}-1)}{4^{s}(4^{s}-2)} + \frac{2^{s+1}-2 + 2^{2s}}{4^{s}}
 \right)
\end{eqnarray}
Which then also has a meromorphic continuation to the complex plane with poles at known locations. In \cite{KKPSS} formulae for the spectral zeta functions of any Laakso space with periodic data $\{j_i\}$ are given as a rational complex valued function times the Riemann zeta function. A feature of the spectral zeta functions on Laakso spaces that does not appear in the interval case is the existence of poles for $\zeta_L(s)$ off of the real axis. These are referred to as complex dimensions \cite{LvF}. In \cite{BegueEtAl2012} the residues of $\zeta_L(s)$ are used to calculate the leading terms of the Weyl asymptotics for $\Delta$ which for Laakso spaces with periodic $j_i$ have a log-periodic oscillating term of leading order.

\begin{theorem}
Suppose that $F_{\infty}$ is constructed in such a way that $F_0 \subset \B{R}^{d}$ is compact with non-empty interior and $\phi_0(B_i)$ induces a self-similar cell structure on $F_0$. Further assume that the $B_i$ have empty interior. Then the associated spectral zeta function will have a tower of simple poles above the spectral dimension. Furthermore, if the spectral zeta function over $F_0$ is meromorphic on the entire complex plane so is the spectral zeta function over $F_{\infty}$.
\end{theorem}

\begin{proof}
As shown in \cite{SteinhurstTeplyaevA} the projective limit construction in this case will be a non-negative real and discrete spectrum. Because it is assumed that $F_{1}$ can be realized as an assembly of identical pieces that are scaled copies of $F_0$ that overlap only on $B_1$ which has empty interior then $Dom(\Delta_1)  = \phi^{*}_1(Dom(\Delta_0)) \oplus \mathfrak{F}_1$ where $\mathfrak{F}_1$ are the eigenfunctions that are orthogonal to $\phi^{*}_iDom(\Delta_0)$ in $L^{2}(F_1)$. By a geometrical argument these eigenfunction are piece-wise defined as eigenfunctions on scaled copied of $F_0$ with suitable matching conditions to assure the orthogonality. See \cite{RS2009} for the case of Laakso spaces. Thus $\sigma(\Delta_i|_{\mathfrak{F}_1}) = c_1\sigma(\Delta_0)$ for some $c_1$. By the self-similarity of the cell structure the constant $c_1$ is the same for all $n$ not just $n=1$. This gives rise to a geometric series over $n$ whose summation has a series of simple poles over the spectral dimension of $L$. Such series will be the topic of Section \ref{sec:finite}. Another series of eigenfunctions could occur due to Neumann boundary conditions but these will also have the same scaling and will merely provide another tower of simple poles. 

The spectral zeta function over $F_{\infty}$ is the sum over $n$ of scaled copied of the spectral zeta function over $F_0$ plus a finite number of bootstrap terms. Because this is actually the same geometric summation as in the previous paragraph the sum is a meromorphic function after regularization if and only if each term is meromorphic.
\end{proof}

\section{Casimir Effect}\label{sec:CE}
The Casimir effect arising between conductors and the quantum vacuum can be viewed as a consequence of vacuum zero-point energy. Simply put, 
displacing conductors generates new boundary conditions for the quantized vacuum, which in turn alters the zero-point energy and gives rise to a negative energy gradient. It is experimentally verified \cite{Sparnaay1958} that two parallel uncharged conducting plates experience an attractive pressure given by

\begin{equation*}
|P_C|=\frac{ \pi^2 c \hbar}{240 d^4}.
\end{equation*}

That such an attraction has its origins in relativistic quantum mechanics is reflected by the appearance of both Planck's quantum mechanical $\hbar$ and the relativistic $c$. However, the direction of the Casimir pressure generally depends on the geometry of the conductors with which one is working. While in the case of plates and cylinder the attraction is positive, both spherical shells and Laakso spaces exhibit repulsion \cite{Canaguier2011}.

In computing the the Casimir effect on Laakso spaces, we take the vacuum expectation of a self-adjoint Hamiltonian operator which represents the quantized electromagnetic field and whose spectrum yields the permissible energies for the system. In particular, as our boundary conditions depend on on some displacement parameter $d$, we obtain

\begin{equation*}
E_{vac}(d)=\langle 0 | H(d) | 0 \rangle \propto \sum_{\lambda \in \sigma(\Delta)} \omega_{\lambda}(d) \propto \sum_{\lambda \in \sigma(\Delta)} \sqrt{ \lambda(d)} \propto \zeta_{L(d)} (-1/2).
\end{equation*}

The Casimir pressure will therefore be proportional to a derivative of the Laakso spectral zeta function evaluated at $-1/2$.  In \cite{KKPSS}, the authors looked at the Casimir effect on $j_n=j$ Laakso spaces that arose from conducting plates attached at nodes in the $F_1$ graph approximation and placed symmetrically about the center.  At each point of intersection with L, the conducting plates imposed Dirichlet boundary conditions on eigenfunctions and Kirchoff boundary conditions were maintained at the other nodes.  Following the above outline, we then computed a modified spectral zeta function subject to these new boundary conditions. To make sense of the energy gradient, we allowed the plates to move symmetrically from their original locations compressing and stretching the underlying space in a natural way. Moving plates closer together compressed the interior space and stretched the exterior. Conversely, moving  the plates away from each other stretched the interior space and compressed the exterior.

\section{Finite Approximations to Laakso Spaces}\label{sec:finite}
In this section we consider the Casimir effect on $F_m$ with Laplacian $\Delta_m$ in the case of two perfectly conducting plates placed at opposite ends of the unit interval. These boundary conditions are simply the Dirichlet boundary conditions. By truncating the counting arguments mentioned in Section \ref{sec:laakso} we see that the spectrum of $\Delta_m$ is given by 
\begin{eqnarray}
	\sigma(\Delta_m) &=& \bigcup_{k=1}^{\infty} \{\pi^2 k^2\}\cup\bigcup_{n=1}^{m} \bigcup_{k=0}^{\infty} \{(k+1/2)^2\pi^2 d_n^2\} \cup\bigcup_{n=1}^{m} \bigcup_{k=1}^\infty \{k^2\pi^2d_n^2\} \notag \\ 
	&& \cup\bigcup_{n=2}^{m} \bigcup_{k=1}^\infty \{k^2\pi^2d_n^2\} \cup \bigcup_{n=2}^{m} \bigcup_{k=1}^{\infty} \left\{\frac{k^2\pi^2 d_n^2}{4}\right\}
\end{eqnarray}
With multiplicities
\begin{equation}
	1,\hspace{.5cm} 2^n,\hspace{.5cm} 2^{n-1}(j_n-2)d_{n-1},\hspace{.5cm} 2^{n-1}(d_{n-1}-1),\hspace{.5cm} 2^{n-2}(d_{n-1}-1)
\end{equation}
From which the spectral zeta function for $\Delta_m$, denoted $\zeta_m(s)$, can be easily calculated in the case where $j_i = j$ and $m \ge 4$ to be
\begin{eqnarray}
	\zeta_m(s) &=& \frac{\zeta_{R}(2s)}{\pi^{2s}} \left[ 1 + \frac{1-(2j^{-2s})^{m}}{j^{2s}(1-2j^{-2s})}\left( 2^{1+2s}-3-\frac{2}{1-2j^{-2s}} \right) \right. \notag \\
	&& + \left. \frac{1-(2j^{1-2s})^{m}}{j^{2s}(1-2j^{1-2s})} \left(j + \frac{4^{2s+2}j^{1-2s}}{j^{2s}(1-2j^{1-2s})} \right) \right]
\end{eqnarray}

The observed poles arise from the use of the summation formula for geometric series
\begin{equation}
	\sum_{n=1}^{m} r^{n} = r\frac{1-r^{m}}{1-r},
\end{equation}
which are not removable due to the interactions between the several sums being taken simultaneously. Since the poles found are outside of the domain of convergence for the summation over $k$ these poles can only be approached through an analytic continuation and so we can choose to represent the summations using this formula since the extended functions will agree on an open subdomain and hence everywhere. 

\begin{prop}\label{prop:finiteperiodic}
For all $T-$periodic sequences $\{j_i\}$ the spectral zeta function corresponding to $\Delta_n$ on $F_n$ for $n \ge 3T$ have towers of complex poles. 
\end{prop}

\begin{proof}
It has already been seen for constant sequences $j_i$. For periodic sequences with longer periods, the summation methods are augmented by summing over an individual period then summing over all periods. The requirement that $n$ is large enough for three periods is so that the geometric aspect of the summation over $n$ is fully present. 
\end{proof}

\section{Casimir Effect on L}\label{sec:CasimirL}

This is the original setting in which the authors considered a Casimir effect in \cite{KKPSS}. In addition to the boundary of the Laakso space playing the part of one pair of plates with Neumann boundary conditions we insert another symmetrically placed pair of plates in the interior so that the dependence of Casimir force on the separation of the plates can be explored. In this section we will consider for the sake of simplicity Laakso spaces with constant sequences $j_i=j$. The interior plates will be places symmetrically at level one wormhole locations. That is their location in the unit interval will be taken from $L_1$. Set $X_0$ to be one half the distance between the interior plates, this gives the distance of the plates from the ``center'' of the Laakso space. Let $Z$ be the number of nodes between the plates in the $F_1$ graph approximation of $L$. See Figure \ref{fig:Lsix} for an example.

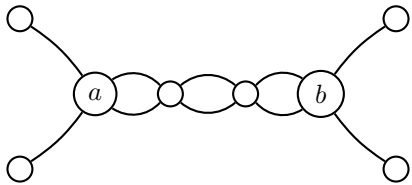
\begin{figure}[t]
\begin{center}
\begin{tikzpicture}[thick, scale=1]
\tikzstyle{every node}=[draw,shape=circle];
\node (11) at (0,2) {};
\node (12) at (0,0) {};
\node (21) at (1,1) {$a$};
\node (31) at (2,1) {};
\node (41) at (3,1) {};
\node (51) at (4,1) {$b$};
\node (61) at (5,2) {};
\node (62) at (5,0) {};

\path[-] (11) edge [bend  left=10] (21);
\path[-] (12) edge [bend  right=10] (21);

\path[-] (21) edge [bend  left=40] (31);
\path[-] (21) edge [bend  right=40] (31);

\path[-] (31) edge [bend  left=40] (41);
\path[-] (31) edge [bend  right=40] (41);

\path[-] (41) edge [bend  left=40] (51);
\path[-] (41) edge [bend  right=40] (51);

\path[-] (51) edge [bend left=10] (61);
\path[-] (51) edge [bend right=10] (62);
\end{tikzpicture}\end{center}
\caption{The $F_1$ graph for $j=5$ where the interior plates are placed at nodes $a$ and $b$. Here $X_0 = \frac{3}{5} \times \frac{1}{2} = \frac{3}{10}$ and $Z = 3$.}
\label{fig:Lsix}
\end{figure}

\begin{definition}
Given a Laakso space with $j_i = j$ and a symmetrically located pair of plates whose location is determined by $j$ and a chosen $Z$. The operator $\Delta'$ on $L^{2}(L,\mu)$ acts as $\Delta$ but with domain determined by imposing Dirichlet boundary conditions at the interior plates. 
\end{definition}

\begin{theorem}[\cite{KKPSS} Theorem 4.2]
The operator $\Delta'$ is self-adjoint and has spectrum with multiplicities as given in Figure \ref{fig:sigmaDelta'}
\end{theorem}

\begin{figure}[t]
\begin{eqnarray*}
	\sigma(\Delta') = && \bigcup_{k=1}^\infty \left\{\left[ \frac{ k \pi}{2X_{0}} \right]^{2}\right\} \cup  \bigcup_{k=0}^\infty \left\{\left[ \frac{ (k+1/2) \pi}{(1-2X_{0})/2} \right]^{2}\right\} \notag \\
	& \cup & \bigcup_{k=0}^\infty \left\{\left[(k+1/2) \pi \frac{j-(Z+1)}{1-2X_{0}} \right]^{2}\right\}\cup \bigcup_{k=1}^\infty \left\{\left[k \pi \frac{j-(Z+1)}{1-2X_{0}} \right]^{2}\right\} \notag \\ 
	&\cup & \bigcup_{k=1}^\infty  \left\{\left[ k \pi \frac{Z+1}{2X_{0}} \right]^{2}\right\} \cup \bigcup_{n=2}^{\infty}  \bigcup_{k=0}^\infty \left\{\left[ I_{n}(k+1/2) \pi \frac{1-\frac{Z+1}{j}}{1-2X_{0}} \right]^{2}\right\} \notag \\ 
	&\cup & \bigcup_{n=2}^{\infty} \bigcup_{k=1}^\infty \left\{\left[I_{n} k \pi \frac{1-\frac{Z+1}{j}}{1-2X_{0}} \right]^{2}\right\}  \cup \bigcup_{n=2}^{\infty} \bigcup_{k=1}^\infty \left\{\left[ I_{n} k \pi \frac{1-\frac{Z+1}{j}}{2(1-2X_{0})} \right]^{2}\right\} \notag \\ 
	&\cup &  \bigcup_{n=2}^{\infty} \bigcup_{k=1}^\infty \left\{\left[I_{n}k \pi \frac{Z+1}{2jX_{0}} \right]^{2}\right\} \cup \bigcup_{n=2}^{\infty} \bigcup_{k=1}^\infty \left\{\left[ I_{n}k \pi \frac{Z+1}{4jX_{0}} \right]^{2}\right\}. 
\end{eqnarray*}
and multiplicities are listed in the same order
\begin{enumerate}
	\item[1)] $1$;
	\item[2)] $2$;
	\item[3)] $2$;
	\item[4)] $j-Z-3$;
	\item[5)] $Z+1$;
	\item[6)] $2^{n}$;
	\item[7)] $(1-\frac{Z+1}{j})I_{n-1}2^{n-1}(j-2)+2^{n-1}(1-\frac{Z+1}{j})I_{n-1}$;
	\item[8)] $2^{n-2}[(1-(\frac{Z+1}{j})I_{n-1}-1]-2^{n-2}$;
	\item[9)] $\frac{Z+1}{j}I_{n-1}2^{n-1}(j-2)+2^{n-1}\frac{Z+1}{j}I_{n-1}+2^{n-1}$;
	\item[10)] $2^{n-2}[\frac{Z+1}{j}I_{n-1}-1]$.
\end{enumerate}
\caption{}
\label{fig:sigmaDelta'}
\end{figure}

Using this it is a tedious but straightforward task to calculate $\zeta_{\Delta'}(s)$. Since we will be interested in how $\zeta_{\Delta'}(s)$ varies as $j$, $Z$, and $X_0$ are varied we will make the dependence explicit by writing $\zeta_{\Delta'}(s) = \zeta_{j,X_0,Z}(s)$.

\begin{cor}
Given a Laakso space with $j_i =j$, and conducting plates placed according to $X_0$ and $Z$ we have the spectral zeta function given in Figure \ref{fig:zetajXZ}.
\end{cor}

\begin{figure}[t]
\begin{eqnarray*}
\zeta_{j,X_{0}, Z}(s) =&&  \sum _{k=0}^\infty \frac{2}{[(2k+1) \pi/(1-2X_{0})]^{2s}} +\sum_{k=1}^\infty \frac{1}{[k \pi/(2X_{0})]^{2s}} \notag \\
	& + & \sum _{k=1}^\infty \frac{(j-Z-3)}{[jk \pi\frac{(1-(Z+1)/j)}{1-2X_{0}}]^{2s}}+\sum _{k=1}^\infty \frac{Z+1}{[k \pi(Z+1) / (2X_{0})]^{2s}} \notag \\ 
	& + & \sum_{n=1}^\infty \sum_{k=0}^\infty \frac{2^{n}}{[I_{n}(k+1/2) \pi (1-\frac{Z+1}{j})/(1-2X_{0})]^{2s}} \notag \\ 
	& + & \sum_{n=2}^\infty \sum_{k=1}^\infty \frac{(1-\frac{Z+1}{j})2^{n-1}I_{n-1}(j-2)+2^{n-1}(1-(Z+1)/j) I_{n-1}}{[I_{n}k \pi \frac{(1-(Z+1)/j)}{(1-2X_{0})}]^{2s}} \notag \\
	& + & \sum_{n=2}^\infty \sum_{k=1}^\infty \frac{ 2^{n-2}[(1-(Z+1)/j)I_{n-1}-1]-2^{n-2}}{ [ I_{n}k \pi(1-(Z+1)/j)/[2(1-2X_{0})]]^{2s}} \notag \\  
	& + & \sum_{n=2}^\infty \sum_{k=1}^\infty \frac{(Z+1)/j[2^{n-1}I_{n-1}(j-2)+2^{n-1}I_{n-1}]+2^{n-1}}{[k \pi I_{n}(Z+1)/(2jX_{0})]^{2s}} \notag \\  
	& + & \sum_{n=2}^\infty \sum_{k=1}^\infty  \frac{(Z+1)/j[2^{n-2}I_{n-1}]-2^{n-2}}{[I_{n}k \pi / (4jX_{0})]^{2s}}.
\end{eqnarray*}
\caption{}
\label{fig:zetajXZ}
\end{figure}

\begin{figure}[t]
\includegraphics[scale=1]{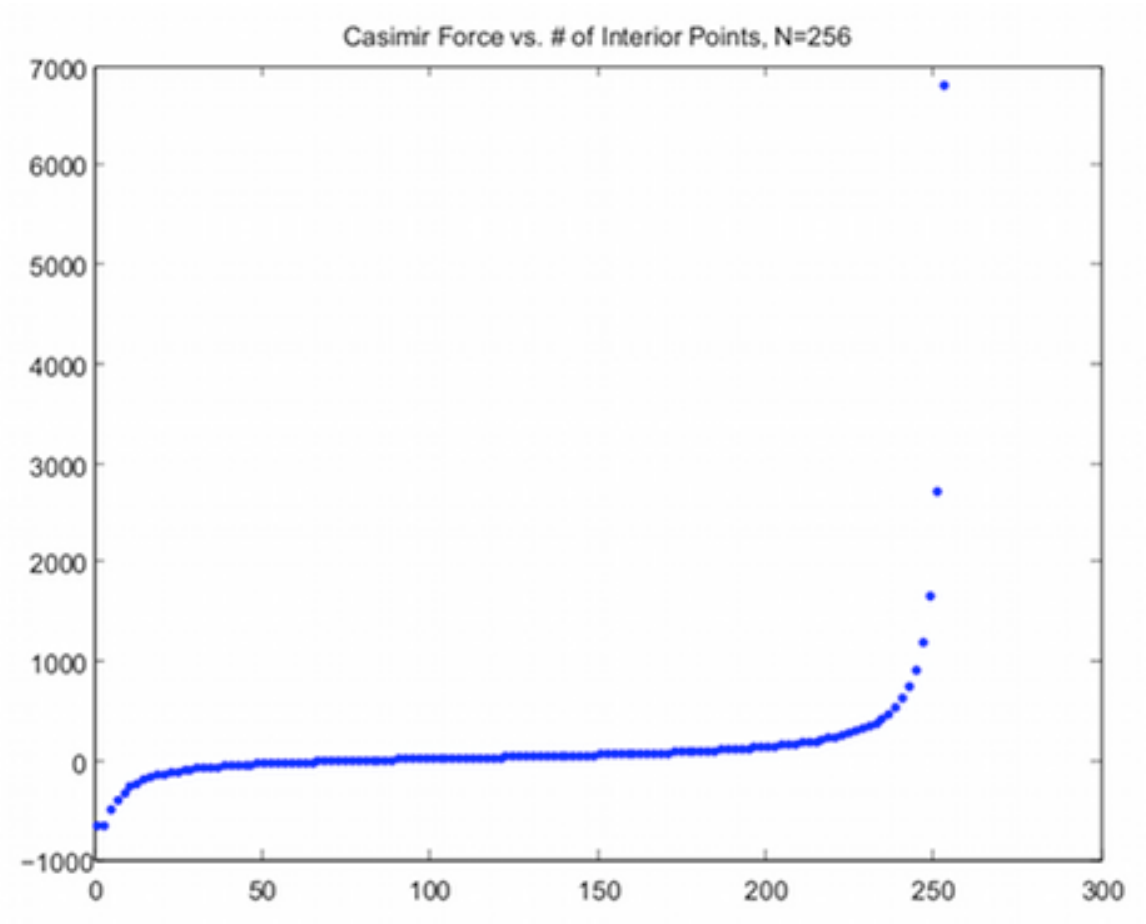}
\caption{Values of the Casimir force plotted for $j=256$ and $Z$ ranging between $1$ and $125$. Notice that for large $Z$ we see an interaction between the plates and the boundary of the Laakso space while for small $Z$ see a similar interaction between the two plates.}
\label{fig:N256}
\end{figure}

\begin{prop}
The Casimir energy of a Laakso space given by $j_i=j$ and plates positioned according to $X_0$ and $Z$ is proportional to $\zeta_{j,X_0,Z}(-1/2)$ and the self-exerted Casimir force due to this energy is proportional to Figure \ref{fig:CasForce}.
\begin{figure}[t]
\begin{eqnarray*}
	F_C(j,Z) & \propto & \frac{d}{dx} \zeta_{j,x,Z}\left( -\frac{1}{2} \right) |_{x=X_0} \\
	&=& \frac{ (j-(Z+1))}{24(1-2j)(1-2X_{0})^{2}} - \frac{ (j-(Z+3))(j-(Z+1))}{12(1-2X_{0})^{2}} \\
	&& - \frac{ j^{3} (j-2)(1-\frac{Z+1}{j})^2}{12 (1-2j^{2})(1-2X_{0})^2}
-\frac{(1-\frac{Z+1}{j})j^{2}(j-(Z+1))}{24 (1-2X_{0})^{2}(1-2j^{2})} \\
	&& +\frac{(1-\frac{Z+1}{j})j^2}{24 (1-2X_{0})^{2}(1-2j)} + \frac{(Z+1)^{2}}{48 X_{0}^{2}} + \frac{(Z+1)^{2} (j-2)}{24 X_{0}^{2} (1-2j^{2})} \\
	&& + \frac{ j (Z+1)^{2}}{96 (1-2j^{2}) X_{0}^{2}} + \frac{1}{6 (1-2X_{0})^{2}} - \frac{ j^{2} (1-\frac{Z+1}{j})}{24 (1-2j) (1-2X_{0})^{2}} \\
	&&+\frac{ j (Z+1)}{96 X_{0}^{2}(1-2j)}-\frac{j^{2} (1-\frac{Z+1}{j})}{12 (1-2X_{0})^{2}(1-2j)} + \frac{1}{48 X_{0}^{2}} + \frac{(Z+1) j}{48 X_{0}^{2} (1-2j)}
\end{eqnarray*}
\caption{}
\label{fig:CasForce}
\end{figure}
\end{prop}

\begin{proof}
See \cite{KKPSS} for details. 
\end{proof}

To see how Casimir force can vary with $Z$ for a given $j$ see Figure \ref{fig:N256}. Inspecting the expression in Figure \ref{fig:CasForce} it is readily apparent that the force depends on the plate separation $X_0$ as $a^{-2}$ instead of the expected $a^{-d_s-1}$. The parameter $Z$ represents how many cells separate the plates, a sort of geometric distance. The dependence on $Z$ is easily seem to be quadratic in Figure \ref{fig:CasForce}. Recall that Laakso spaces all have walk dimension $d_w=2$ so $d_s = d_h = 1+\frac{log(r)}{\log(2)}$. The reason this happens is that $\Delta$ and consequently $\Delta'$ are truely one-dimensional operators rather than $d_s$-dimensional operators so the rapid grown in the eigenvalue counting function is due to the geometry of Laakso spaces rather than dimensionality of the Laplacian.

\section{A Higher Dimensional Case}\label{sec:CasimirLR2}

\subsection{The $3+$ dimensional model}
 Let $L_j$ be the Laakso space represented by the sequence $j_n=j~\forall n$, let K be the Cantor set, and let $I=[0,1]$. We modify the configuration in \cite{KKPSS} by considering the space
 
 \begin{equation*}
 L_j \times \mathbb{R}^2=[( I \times K)/\sim] \times \mathbb{R}^2.
 \end{equation*}
Attach two conducting plates $P_1, P_2 \subset L \times \mathbb{R}^2$ where 
  
  \begin{equation*}
  P_1 =[(0 \times K)/\sim] \times \mathbb{R}^2; P_2=[(1 \times K) / \sim] \times \mathbb{R}^2.
  \end{equation*}
  
  Now allow symmetric displacement of the plates in the I direction of the Laakso space. By construction, plates moving towards one another will compress the interior space between them while plates moving apart will stretch it. 

\subsection{Casimir Effect Strength versus Plate Separation Distance} 
The self-adjoint operator $\Delta_{L_j \times \mathbb{R}^2}$ takes the form 

\begin{equation*}
\Delta_{L_j \times \mathbb{R}^2}=\Delta_{L_j} -\frac{ \partial ^2}{\partial x_1^2} - \frac{\partial^2}{\partial x_2^2} 
\end{equation*}
where $\Delta_{L_j}$ is the non-negative definite self-adjoint Laplacian on $L_j$ such that $\Delta \Phi_n^* f=\Phi_n^* \Delta_n f$ for all $f \in Dom(\Delta_n)$.  The generalized spectrum for $\Delta_{L_j \times \mathbb{R}^2}$ is conveniently described by 

\begin{equation*}
\sigma_{L_j \times \mathbb{R}^2} = \{ \lambda_s + k_{x_1}^2 + k_{x_2}^2 : \lambda_s \in \sigma_{L_j}\}.
\end{equation*}
 
As mentioned in Section \ref{sec:CE}, the zero-point energy for our system arises from the ground state expectation of the Hamiltonian representing the quantized electric field, which takes the form 
\begin{equation}\label{eq:CasimirEnergy}
E_{Cas}=2 \cdot \frac{\hbar}{2} \sum_{\lambda_s , k_{x_1}, k_{x_2}} \omega_{\lambda_s, k_{x_1}, k_{x_2}}= \hbar c \sum_{ \lambda_s, k_{x_1}, k_{x_2}} \sqrt{ \lambda_s +k_{x_1}^2+ k_{x_2}^2}.
\end{equation}
The quantity $\hbar \omega_{\lambda_s, k_{x_1}, k_{x_2}}=\hbar c \sqrt{\lambda_s+k_{x_1}^2+k_{x_2}^2}$ is the electromagnetic energy associated to radiation with wave vector $(\sqrt{\lambda_s}, k_{x_1}, k_{x_2})$. 

We recast this last expression as an integral and switch to polar coordinates:
\begin{eqnarray*}
\frac{E_{Cas}}{Area} &=& \frac{\hbar c}{4 \pi^2}\sum_{\lambda_s \in \sigma_{L_j}} \int_{\mathbb{R}^2} dk_{x_1} dk_{x_2} \sqrt{ \lambda_s + k_{x_1}^2 + k_{x_2}^2} \\ &=& \frac{\hbar c}{4 \pi^2} \sum_{\lambda_s \in \sigma_{L_j}}  \int_0^\infty \int_0^{2\pi} dk \cdot k \sqrt{k^2+\lambda_s} ~d\theta~dk. 
\end{eqnarray*}
We regularize this last quantity to obtain
\begin{equation*}
\frac{E_{Cas}}{Area} = \frac{ \hbar c}{6 \pi} \sum_{ \lambda_s \in \sigma_{L_j}} (\lambda_s)^{3/2}.
\end{equation*}

\begin{prop}
Two conducting plates attached at the boundary of $L_j \times \mathbb{R}^2$ as described in the set up will experience a pressure given by 
\begin{eqnarray*}
	P_{Cas}(j)&=&\frac{ \hbar c \pi^2}{240} \left[ 1+ \frac{ 2j^4}{1-2j^3} + \frac{17}{8} \cdot \left( \frac{j^7}{1-2j^4} - \frac{j^6}{1-2j^3} \right) +  \right. \\
	&& \hspace{1cm} \left.  \left( \frac{j^4}{1-2j^4}- \frac{2j^3}{1-2j^4} \right) \right] \\
	&=& \frac{\hbar c \pi^2}{240} \left[ \frac{ 8 -16j^3 -8j^4 + 15 j^6 + j^7 }{8 - 16j^3 -16 j^4 + 32 j^7} \right],
\end{eqnarray*}
where a positive signed pressure indicates a repulsive force. 
\end{prop}

\begin{proof}
Computing the Casimir pressure involves eigenvalue counting arguments similar to the ones made in \cite{RS2009} except that Neumann boundary conditions are replaced by Dirichlet boundary conditions at the boundary $(0 \times K)/\sim$ and $(1 \times K )/ \sim$. We construct the full spectrum of $\Delta_{L_j}$ by exploiting orthogonality relations between eigenfunctions in different quantum graph approximations.  Decomposing each quantum graph $F_n$ into loops, V's, and crosses, note that the V's are the only shapes whose localized eigenfunctions are altered by the new boundary conditions. Once the new spectrum with multiplicities is found, substitute into the expression for Casimir energy, regularize the sum, and take derivatives with respect to displacement to obtain the result.  
\end{proof}

In the limit $j \rightarrow \infty$ the Hausdorff dimension of the Laakso spaces $L_j$ go to $1$. So in the $j \rightarrow \infty$ limit we might expect something similar to the classical case since the localized eigenfunctions have eigenvalues increasing without bound so in the limit effectively vanish. Simply from the formula we have that $\lim_{j \rightarrow \infty} P_{Cas}(j) = \frac{1}{32} \frac{\hbar c \pi^2}{240}$. Hence, in the limit, our fractal pressure is precisely $1/32$ the magnitude of the classical Casimir value and acting in the opposite direction. 

\begin{cor}
Two conducting plates attached at the boundary of $L_j \times \mathbb{R}^2$ and then stretched to a distance $d$ from one another will experience a Casimir pressure
\begin{equation*}
P_{Cas}(j,d) = \frac{P_{Cas}(j)}{d^4}.
\end{equation*}
\end{cor}

\begin{proof}
This follows immediately from the fact that multiplying the displacement by a factor of $d$ means $\lambda_s \rightarrow \frac{ \lambda_s}{d^2}$ for every $\lambda_s \in \sigma_L$. 
\end{proof}

In particular, the power law governing the Casimir pressure as a function of displacement is independent of the spectral dimension of the Laakso space $L_j$.

\subsection{Casimir Pressure as a function of $\{j_i\}$}
\begin{prop}\label{prop:CasimirPressure2}
Let L is the Laakso space represented by some N-periodic sequence $\{j_i\}$.  Then two conducting plates attached at the boundary of $L \times \mathbb{R}^2$ yield an unnormalized Casimir energy given by 

\begin{eqnarray*}
	E_{Cas} (\{j_i\}, d )   &=& \sum_{k=1}^\infty \left( \frac{k \pi}{d} \right) ^3  + \sum_{n=1}^\infty \sum_{k=1}^\infty \left( \frac{ k \pi I_n}{d} \right) ^3\\
		&& + \sum_{n=1}^\infty \sum_{k=1}^\infty 2^{n-1} I_{n-1} ( j_n-2) \left( \frac{k \pi I_n}{d} \right) ^3 \\
		&& + \frac{17}{16} \sum_{n=2}^\infty \sum_{k=1}^\infty 2^{n-1}\left( I_{n-1} -1 \right) \left( \frac{ k \pi I_n}{d} \right) ^3. 
\end{eqnarray*}
\end{prop} 

\begin{proof}
This more general case follows from Equation \ref{eq:CasimirEnergy} and eigenvalue counting arguments similar to the ones made in \cite{RS2009} for periodic Laakso spaces. 
\end{proof}

\begin{prop}
Let L be the Laakso space represented by some N-periodic sequence $\{j_i\}$. Then two conducting plates attached at the boundary of $L \times \mathbb{R}^2$ will experience a Casimir pressure given by

\begin{eqnarray*}
P_{Cas} (\{j_i\}, d) &=&\frac{\hbar c \pi^2}{240 d^4} \left[ 1+ \frac{15}{32} \left( \sum_{i=1}^N  \prod_{k \leq i} 2 j^3_k  \right) \left( \frac{1}{1-r^{3N}2^N} \right)\right. \\
	&&\left. + \frac{1}{2} \left( \sum_{i=1}^N 2j_k^4 \right)\left( \frac{1}{1-r^{4N}2^N} \right)  \right. \\ 
	&& - \left. \frac{15}{32} \sum_{i=1}^N \left( \prod_{k \leq i} 2j_k^4 \right) \left( \frac{j_i^{N-i} }{j_i^N -r^{4N}2^N} \right) \right]
\end{eqnarray*}
where a positive sign indicates a repulsive force. 
\end{prop}

\begin{proof}
Note that one formally has
\begin{equation*}
\sum_{n=1}^\infty I_n = \left( \sum_{l =0}^\infty r^{Nl} \right) \left( \sum_{i=1}^N \prod_{k \leq i} j_k \right) .
\end{equation*}
Hence, we can regularize 
\begin{eqnarray*}
\sum_{n=1}^\infty \sum_{k=1}^\infty 2^n \left( \frac{ k \pi I_n }{d} \right) ^3 &= & \frac{\zeta(-3) \pi ^3 }{d^3} \sum_{n=1} 2^n I^3_n \\&=& \frac{\zeta(-3) \pi ^3 }{d^3}  \sum_{l=0}^\infty r^{3Nl} \left( \sum_{i=1}^N 2^{ Nl+i}\prod_{k \leq i} j_k^3 \right) \\&=& \frac{\zeta(-3) \pi ^3 }{d^3} \left( \sum_{l=0}^\infty r^{3Nl} 2^{Nl} \right) \left( \sum_{i=1} ^N \prod_{k \leq i} 2j_k^3 \right) \\&=& \frac{\pi^3}{120 d^3} \left( \frac{1}{1-r^{3N} 2^N} \right) \left( \sum_{i=1}^N \prod_{k \leq i} 2j_k^3 \right). 
\end{eqnarray*}
The unregularized Casimir energy for $L \times \mathbb{R}^2$ is comprised of terms that look quite similar to above sum.  The reader may obtain the formula for Casimir pressure by first regularizing each term of the Casimir energy in  Proposition \ref{prop:CasimirPressure2}  and then taking a (negative) spatial derivative with respect to d. 
\end{proof}

\small
\bibliography{laaksocasimir}{}
\bibliographystyle{plain}

\end{document}